\documentclass{amsart}
\usepackage{amsmath, amssymb}
\newtheorem{Theorem}{Theorem}[section]

\newtheorem{Proposition}[Theorem]{Proposition}
\theoremstyle{definition}
\newtheorem{Definition}[Theorem]{Definition}
\theoremstyle{remark}

\numberwithin{equation}{section}

\newcommand{\R}{\mathbb R}

\begin{document}

\title{A mathematical bridge between discretized gauge theories in quantum physics and approximate reasonning in pairwise comparisons}
\author{Jean-Pierre Magnot}

\address{LAREMA, Universit\'e d’Angers, 2 Bd Lavoisier
	, 49045 Angers cedex 1, France and Lyc\'ee Jeanne d'Arc, 40 avenue de Grande Bretagne, 63000 Clermont-Ferrand, France}%
\email{ jean-pierr.magnot@ac-clermont.fr}

\begin{abstract}
	We describe a mathematical link between aspects of information theory, called pairwise comparisons, and discretized gauge theories. The link is made by the notion of holonomy along the edges of a simplex. This correspondance leads to open questions in both field.
\end{abstract}
\maketitle

\section*{Introduction}
We present here an overview, addressed to physicists, of a possible bridge between gauge theories (and also some aspects of quantum gravity) with pairwise comparisons matrices and their applications in information theory and approximate reasonning. This is the reason why we fastly summarize features in physics (assuming that they are known by the reader), and give more details for selected features on Pairwise Comparisons (PC) matrices (assuming that this field is less known). This paper is a companion work to \cite{Ma-K,Ma-alg,Ma2017} where, after a tentative in \cite{KSW2016}, the extension of the notion of classical PC matrices to matrices with coefficients in a group is considered, partially motivated (in my case) by the striking similarities with mathematical constructions in discretized gauge theories. A not complete list of reference about PC matrices is \cite{DK1994,K1993,KS2014}, oriented in our perspective, and a very partial list of references about gauge theories and their discretized forms is \cite{AHK,AHKH1,AHM,AZ1990,Hah2004,HCR2015,Lim2012,Re1997,RV2014,SSSA2000,Sengupta,Sengupta2004}. 

We begin with an oriented survey of selected problems in  discretization of $G-$gauge theories, where $G$ is a Lie group, and a selection of features in evaluation of inconsistency in pairwise comparisons with coefficients in $\R_+^*.$ Then we describe, following \cite{Ma-K,Ma-alg}, a straightforward extension of  PC matrices with coefficients in $\R_+^*$ to a general Lie group $G.$ The link with gauge theories is performed via holonomy, which appears in discretizations described in \cite{Ma2017,RV2014}. We finish with the possible interpretations in both sides of this correspondence, first from quantities on PC matrices to gauge theories, and secondly from second quantization to approximate reasonning. 
 \section{A short and not complete survey of each field of knowledge}
 We present here the two fields under consideration, discretized gauge theories and pairwise comparisons in aproximate reasonning, in a way to highlight the correspondances. 
 \subsection{Gauge theories discretized}
 The phase space of a (continuum) gauge theory is a space of connections on a (finite dimensional) principal bundle $P$ with structure group $G$ and with base $M$. We note by $C(P)$ the space of connections considered. If $M$ is not compact and Riemannian, one often use the space of connections which are smooth and square-integrable. A gauge theory is defined by an action functionnal $S:C(P)\rightarrow \R$ which has to be minimized. 
 
 A discretized gauge theory is defined on a triangulation, a cubification or any other way to discretize the manifold $M,$ and the principal bundle $P$ can be often trivial. Let us highlight two kind of discretizations:
 
 \begin{itemize}
 	\item Whitney's discretization \cite{Wh}, where a connection $\theta \in C(P)$ is (Riemann-)integrated on the 1-vertices of the chosen triangulation. This mimiks a finite element method of approximation for scalar functions, and the discretized connection $\theta_W$ generates a $H^1-$approximation $\tilde\theta$ of $\theta,$ which $H^1-$converges to $\theta$ when refining the triangulation. The action functionnal $S$ is then evaluated on the finite dimensional space of cnnnections $\tilde{\theta}$ instead of the infinite dmensional space $C(P).$ This is, to our knowledge the most widely developped approach, but this approach seems to fail partially for non-abelian theories, partly because the triviaisation is not gauge-covariant. Ths leads to gauge-fixing strategies.
 	\item Holonomy discretization, mostly inspired by the ideas of quantum gravity \cite{RV2014}, where connections along the edges are discretized through their holonomy. This approach requires mathematical precisions by fixing a preliminary gauge on the 1-vertices of the discretized manifold, but in a final analysis, only theories depending of secondary characteristic classes (e.g. Chern-Simons theory) can give rise to pathologies in gauge covariance, where as gauge theories arising from primary characteristic classes (e.g. Yang-Mills theories) are fully gauge covariant when discretized \cite{Ma2017}. 
 \end{itemize}  
 \subsection{Pairwise comparisons, consistency and inconsistency}
 It is easy to explain the inconsistency in PC matrices when we consider cycles of three comparisons, called triads and represented here as $(x,y,z)$,
 which do not have the ``morphism of groupoid'' property such as $$x.z \neq y$$ 

 \noindent The use of ``inconsistency'' has a meaning of a measure of inconsistency in this study; not the concept itself. 
 One approach to inconsistency (originated  in \cite{K1993} and generalized in \cite{DK1994}) can be reduced to a simple observation:
 \begin{itemize}
 	\item 
 	search all triads (which generate all 3 by 3 PC submatrices) and locate the worst triad with an inconsistency indicator ($ii$),
 	\item 
 	$ii$ of the worst triad becomes $ii$ of the entire PC matrix.
 \end{itemize}
 
 \noindent Expressing it a bit more formally in terms of triads (the upper triangle of a PC submatrix $3 \times 3$), we have:
 $$ii_3(x,y,z) = 1-\min\left \{\frac{y}{xz},\frac{xz}{y}\right \} =1- e^{-\left|\ln\left (\frac{y}{xz}\right )\right |}$$
 \noindent The expression $|\ln(\frac{y}{xz})|$
 stands for the distance of the triad $T$ to the "nearest" consistent PC matrix. When this distance increases, the $ii(x,y,z)$ also increases. It is important to notice here that this definition allows us to localize the inconsistency in the PC matrix, which is of a considerable importance for most applications. For highr rank matrices, $ii_3$ is evaluated on each $3\times 3$ matrices, taking the supremum of the obtained values. 
 
 Another possible definition of the inconsistency indicator can also be defined (following \cite{KS2014}) as:
 $${\rm ii}_n(A)=1-\min_{1\le i<j\le n}  \min\left ({a_{ij}\over a_{i,i+1}a_{i+1,i+2}\ldots a_{j-1,j}},\,
 {a_{i,i+1}a_{i+1,i+2}\ldots a_{j-1,j}\over a_{ij}} \right ) $$
 
 \noindent since the matrix $A$ is consistent if and only if for any $1\le i<j\le n$ the following equation holds:
 $$a_{ij}=a_{i,i+1}a_{i+1,i+2}\ldots a_{j-1,j}.$$
 
 \noindent This is equivalent to:
 $${\rm ii}_n(A)=1-\max_{1\le i<j\le n}
 \left (1 - e^{-\left |\ln \left ( {a_{ij}\over
 		a_{i,i+1}a_{i+1,i+2}\ldots  a_{j-1,j}}  \right )\right |}\right )
 $$
 
 One of the main features in applications is to minimize the inconsistency indicator $ii.$ This is the reason why, instead of using the inconsistency indicator $ii$ defined before, there is plethore of inconsistency indicators. Each inconsistency indicator intends to measure how far a PC matrix is from the set of consistent PC matrices, which, for $3 \times 3$ PC matrices, is a 2-dimensional manifold of matrices of the form $$\left(\begin{array}{ccc}
 1 & x & xy \\ x^{-1} & 1 & y \\ x^{-1}y^{-1} & y^{-1} & 1
 \end{array}\right), \hbox{ with } (x,y)\in (\mathbb{R}_+^*)^2.$$  
Unfortunately, the notion and the theory of inconsistency indicators is not actually fixed and acheived, and many competing, uncompatible approaches are actually developped.  
 \section{The matrix of holonomies and pairwise comparisons}
 
 We follow first \cite{Ma-alg}.
 Let $I$ be a set of indexes among $\mathbb{Z},$ $\mathbb{N}$ or $\{0,...,n\} $ for some $n \in \mathbb{N}^*.$
 
 \begin{Definition}
 	Let $(G,.)$ be a group. A PC matrix is a matrix $$A= (a_{i,j})_{(i,j)\in I^2} $$
 	such that:
 	
 	\begin{enumerate}
 		\item $\forall(i,j) \in I^2, a_{i,j}\in G.$

 		\item  \label{inverse} $ \forall (i,j)\in I^2, \quad a_{j,i}=a_{i,j}^{-1}. $
 		
 		\item \label{diagonal} $a_{i,i} = 1_G.$ 
 		
 	\end{enumerate}
 	
 \end{Definition}
 
 \noindent
 
 The matrix $A$ is \textbf{covariantly consistent} if \begin{equation}
 \label{consistent} \forall (i,j,k)\in I^3, \quad a_{i,j}.a_{j,k} = a_{i,k}.
 \end{equation}
 
 Due to the contravariant composition thereafter, we  will use the following notion: the PC matrix is \textbf{contravariantly consistent} if
 \begin{equation}
 \label{consistent2} \forall (i,j,k)\in I^3, \quad a_{j,k}. a_{i,j} = a_{i,k}.
 \end{equation} 
 These two notions are dual, and depend on which order we require for the group multiplication. 
 Similarly, a contravariant consistent PC matrix $A= (a_{i,j})_{(i,j)\in I^2} $ generates a covariant consistent PC matrix $B= (b_{i,j})_{(i,j)\in I^2} $ setting $$ b_{i,j} = a_{i,j}^{-1}= a_{j,i}.$$
 For convenience, we use the terms \textbf{covariant PC matrix} (resp. \textbf{contravariant PC matrix} ) when covariant consistency (resp. contravariant consistency) is naturally required.
 
 \begin{Theorem} \label{th1}
 	\begin{equation*}
 	\exists (\lambda_i)_{i \in I}, \quad a_{i,j} = \lambda_i^{-1}. \lambda_j \Leftrightarrow A ~~is~consistent.
 	\end{equation*}  
 	
 \end{Theorem}
Let $n \in \mathbb{N}^*$
and $$\Delta_n = \left\{ (x_0,\ldots, x_n) \in \mathbb{R}^{n+1} | \left(\sum_{i=0}^n x_i = 1\right) \wedge \left(\forall i \in \{0,\ldots.n\}, x_i \geq 0\right)\right\}$$
be an $n-$simplex.
This simplex can be generalized to the infinite dimension: 
$$\Delta_{\mathbb{N}} = \left\{ ( x_n)_{n \in \mathbb{N}} \in l^1(\mathbb{N}, \mathbb{R}_+^*) | \sum_{i=0}^\infty x_i = 1\right\}$$
\noindent	and 	$$\Delta_{\mathbb{Z}} = \left\{ ( x_n)_{n \in \mathbb{Z}} \in l^1(\mathbb{Z}, \mathbb{R}_+^*) | \sum_{i \in \mathbb{Z}} x_i = 1\right\},$$
where the summation over $\mathbb{Z}$ is done by integration with respect to the counting measure.
In the sequel, $\Delta$ will denote $\Delta_n,$ $\Delta_{\mathbb{N}}$ or $\Delta_{\mathbb{Z}}.$ 
We define a gauge $({g_i})_{i \in I}\in G^I$ with  $\widetilde{\gamma_i}(1) = (\gamma_i(1),{g_i})$ where $$\gamma_i = [s_0,s_1]*\ldots*[s_{i-1},s_i] \hbox{ if } i>0$$
\noindent and  $$\gamma_i = [s_0,s_{-1}]*\ldots*[s_{i+1},s_i] \hbox{ if } i<0.$$
We set $g_i = Hol_{(s_0,1_G)}\gamma_i.$
Let us recall that, for two paths $c$ and $c'$ such that $c*c'$ exists (i.e. $c(1)=c'(0)$), if $p=(c(0),e_G),$ $p'=(c'(0),e_G)$ and $h = Hol_{p}c,$ we have:


\noindent Let \begin{equation} \label{aij} 	a_{i,j} = g_j.Hol_{(s_0;e_G)}\left( \gamma_i [s_i, s_j] \gamma_j^{-1} \right).g_i^{-1}. \end{equation}

In the light of these specifications, we set, for any connection $\theta \in \Omega^1(\Delta, \mathfrak{g}),$

 $$ A = Mat(a_{i,j})$$
and the required notion of consistency is contravariant consistency.

\begin{Proposition}
	$A$ is a PC matrix.
\end{Proposition}

\begin{proof} This follows from holonomy in ``reverse orientation''. 
\end{proof}

\vskip 12pt
\noindent
Let $\gamma_{i,j,k} = \gamma_i * [s_i,s_j]*[s_j,s_k] * [s_k,s_i] *\gamma_i^{-1}$
be the loop based on $s_i$ along the border of the oriented 2-vertex $[s_i,s_j,s_k],$ where $*$ is the composition of paths.
Contravariant consistency seems to fit naturally with flatness of connections: 
\begin{eqnarray*} \forall i,j,k, & \quad & a_{i,k} =  a_{j,k}.a_{i,j} \\
	&\Leftrightarrow & a_{k,i} . a_{j,k} . a_{i,j} = a_{i,i} = 1_G \\
	& \Leftrightarrow & Hol (\gamma_{i,j,k}) = 1_G
\end{eqnarray*}

\medskip
\noindent By fixing an \textbf{indicator map}, defined in \cite{KSW2016} as $$In: G\rightarrow \mathbb{R}$$ to 
$In(1_G)=0,$ we get a generalization of the \textbf{inconsistency indicator} by setting $$ii_{In} = \sup \left\{In(Hol (\gamma_{i,j,k})) | (i,j,k)\in I^3 \right\}.$$
\bigskip
For example, if
$d$ is a left-invariant distance on $G,$ a natural indicator map can be: $$In: g \mapsto d(1_G,g^{-1} ). $$

Such definitions extend to triangularized manifolds along the lines of \cite{Ma2017}, where one can see that the notion of holonomy on a manifold can be discretized, inserting "gaps" (i.e. 0 entries) on a larger matrix  gathering all the PC matrices over the simplexes of the triangulation, along the lines of \cite{Ma-alg}. We then recover the discretization of connections via holonomies decribed in \cite{RV2014}.   
 {When $G=\R_+^*$}
 In (classical) pairwise comparisons, matrix coefficients are scalar, which enable to simplify the settings obtained. The indicator $ii$ appears as the distance of the holonomy of the loop $[s_{i}s_{i+1}...s_{i+n}s_{i}]$ to identity.  
 \section{Correspondances and open questions}
We give here the following table of correspondences, that we comment. We first highlight  
how the notions on PC matrices correcpond to geometric objects. 
\vskip 12pt
\begin{tabular}{|c|c|c|}
	\hline
	PC matrices & discretized gauge theories & continuum gauge theories \\
	\hline
	consistency & 0-holonomy & 0-curvature \\
	\hline
	consistency & 0-holonomy& $\Omega(X,Y)=0$\\
	in $3\times3$ matrices&on the border of&with $X,Y$\\
	& a 2-simplex $\Delta_2$& tangent to $\Delta_2$\\
	\hline
	Koczkodaj's & $\sup d(Hol({\partial \Delta_2}), 1_G)$& $\sup_M||\Omega||$ \\
	inconsitency &when $\Delta_2$&\\
	 indicator $ii$ &is any 2-simplex of&\\
	 & the triangulation &\\
	\hline
	Minimization &?&Minimization \\
	of inconsistency&& of the curvature norm\\
	\hline
\end{tabular}
	
\vskip 12pt	
\noindent
Minimization of inconsistency is an important feature for applications of PC matrices. For decison making, it consists in adapting slighlty the parameters of the studied situation
in order to make "approximately consistent" choices. In order to check if the adapted PC matrix is "approximately consistent", the criteria is given by the chosen inconsistency indicator, e.g. Koczkodaj's $ii,$ which needs to have a value $\epsilon \geq 0 $ small enough. Again for Koczkodaj's $ii,$ it is trivial to see that the sets $$V_\epsilon = ii_{3}^{-1}[0;\epsilon[$$
form a filter base of open neighborhoods of the set of consistent PC matrices. The analog for continuum gauge theories is considering 
$$ U_\epsilon = \left\{\theta \in C(P)|\sup_M||\Omega||<\epsilon\right\}.$$
This defines a filter base of open neighborhoods of the set of 0-curvature connections. This leads to the following open problem.

On one hand, the procedures developped to get consistencizations of PC matrices are specific to the case $G=\R_+^*$. However, there exists many methods for minimizing functionnals, and one may wonder whether minimization of the curvature norm have a physical meaning. The corresponding physical quantity would be the Yang-Mills action functionnal, but there we get here an average value instead of a supremum. This questions the choice of the supremum in the formula $$ii_3(A) = \sup \left\{ii_3(B)| B \hbox{ is a } 3 \times 3 \hbox{ PC submatrix of } A \right\}.$$

On the other hand, considering flat connections, or only connections such that $\Omega=d\theta$ is a common feature in gauge fixing. Since minimal states represent stable solutions, consistencization may appear as a way to equilibrum. 

Let us now reverse the perspective and consider second quantization. We now analyze how the Feynman-like integration (i.e. cylindrical integration for discretized theories) may  arise in PC matrices. 
	\vskip 12pt
	\begin{tabular}{|c|c|c|}
		\hline
	Continuum integration& discretized integration& PC matrices\\
	\hline
	heuristic Lebesgue&Lebesgue measure&?\\
	measure&(finite dimension,&\\
	&Whitney discretization)&\\
	\hline
	integral to be studied&Product measure on $G^{n\times n}$& Measures on $n \times n$ \\
	&(G compact), on the&PC matrices\\
	&discretization by holonomies&\\
	\hline
	
\end{tabular}

\vskip 12pt
If $G$ is compact, we can assume that it is of volume 1. In this case, there is a convergence at the continuum limit to an integral. this si the approach suggested in \cite{RV2014}, which is an alternative approach to the classical integration with respect to the heuristic Lebesgue integral, see e.g. \cite{AHM}. On one hand,  for both cases, the possible interpretations in terms of applications of PC matrices are not investigated. Heuristically, measures on PC matrices may arise when evaluations are random, or subject to measurement errors. On the other hand, PC matrices may furnish an interpretation in terms of information theory of Feynman type integration. 

\vskip 12pt

The author declares that there is no conflict of interest with respect to the publication of this paper.

\end{document}